\documentclass{imsart}
\usepackage{amsthm}
\usepackage{amsmath}
\usepackage{amssymb, amsfonts}
\usepackage{mathtools}
\usepackage{natbib}
\usepackage{bm}
\usepackage{bbm}
\usepackage{graphicx}
\usepackage{courier}
\usepackage{color}
\usepackage{subfig}
\usepackage{floatrow}
\usepackage{placeins}
\usepackage{afterpage}
\usepackage{units}
\usepackage{url}
\usepackage{verbatim}
\usepackage{epstopdf}
\usepackage{booktabs,caption,fixltx2e}
\usepackage[flushleft]{threeparttable}
\usepackage{rotating}
\usepackage{multirow}

\newcommand{\Var}{\mbox{Var}}

\newtheorem{thm}{Theorem}[section]
\theoremstyle{remark}
\newtheorem{remark}{Remark}

\usepackage{enumitem, hyperref, nameref}
\makeatletter 
\def\namedlabel#1#2{\begingroup
    #2%
    \def\@currentlabel{#2}%
    \phantomsection\label{#1}\endgroup
}
\makeatother

\begin{document}

\begin{frontmatter}
\title{Thresholding Nonprobability Units in Combined Data for Efficient Domain Estimation\protect\thanksref{T1}}
\runtitle{Thresholding Nonprobability Units in Combined Data}
\runauthor{T. Savitsky et al.}
\thankstext{T1}{U.S. Bureau of Labor Statistics, 4600 Silver Hill Road , Suitland, MD 20746 USA}

\begin{aug}
\author[A]{\fnms{Terrance D.}~\snm{Savitsky}\ead[label=e2]{Savitsky.Terrance@bls.gov}\orcid{0000-0003-1843-3106}},
\and
\author[B]{\fnms{Matthew R.}~\snm{Williams}\ead[label=e3]{mrwilliams@rti.org}\orcid{0000-0001-8894-1240}}.
\and
\author[C]{\fnms{Vladislav}~\snm{Beresovsky}\ead[label=e1]{Beresovsky.Vladislav@bls.gov}\orcid{0009-0002-8375-5195}},
\and
\author[D]{\fnms{Julie}~\snm{Gershunskaya}\ead[label=e4]{Gershunskaya.Julie@bls.gov}\orcid{0000-0002-0096-186X}}
\address[A]{Office of Survey Methods Research,
U.S. Bureau of Labor Statistics\printead[presep={,\ }]{e1}}

\address[B]{Office of Survey Methods Research,
U.S. Bureau of Labor Statistics\printead[presep={,\ }]{e2}}

\address[C]{RTI International\printead[presep={,\ }]{e3}}

\address[D]{OEUS Statistical Methods Division,
U.S. Bureau of Labor Statistics\printead[presep={,\ }]{e4}}

\end{aug}

\begin{abstract}
Quasi-randomization approaches estimate latent participation probabilities for units from a nonprobability / convenience sample.  Estimation of participation probabilities for convenience units allows their combination with units from the randomized survey sample to form a survey weighted domain estimate.  One leverages convenience units for domain estimation under the expectation that estimation precision and bias will improve relative to solely using the survey sample; however, convenience sample units that are very different in their covariate support from the survey sample units may inflate estimation bias or variance.  This paper develops a method to threshold or exclude convenience units to minimize the variance of the resulting survey weighted domain estimator.  We compare our thresholding method with other thresholding constructions in a simulation study for two classes of datasets based on degree of overlap between survey and convenience samples on covariate support.  We reveal that excluding convenience units that each express a low probability of appearing in \emph{both} reference and convenience samples reduces estimation error.
\end{abstract}


\end{frontmatter}

\section{Introduction}

Declining response rates for randomized survey instruments administered by government statistical agencies \citep{10.1093/jssam/smx019} have encouraged the development of quasi-randomization processes such as those of \citet{2009elliot, 2017elliot, 2021valliant, savitskycombine2023} to allow inclusion of responses derived from a nonrandom convenience sample that includes responses for covariates that overlap those measured by the randomized survey or reference sample.  Directly combining responses for units participating in the convenience sample with those selected into the randomized or reference sample may be expected to induce bias for inference about an underlying latent population, however, because the convenience sample is not generally representative of that population \citep{2010Bethlehem, 2018Meng, 2011VanderWeele_Shpitser}.

Quasi-randomization methods propose model formulations to estimate the convenience sample unit marginal participation probabilities as if the convenience sample is realized from a \emph{latent} or unknown selection process.   Quasi-randomization uses the reference sample and associated known inclusion probabilities to provide information about the underlying sampling frame that is, in turn, used to estimate convenience sample inclusion probabilities.  The goal in using a statistical model to estimate the convenience sample inclusion probabilities is to allow inclusion of the convenience sample units in a combined (reference and convenience sample) data estimator for a domain mean (e.g., employment for computer services in New York city) with minimal bias.  

\citet{beresovsky2024review} provides a comprehensive overview of quasi-\newline randomization methods and compares the variance performances of a collection of methods for domain estimation that are mostly differentiated by assumptions about the degree of overlap in memberships in the convenience and reference samples, on the one hand, and the form of approximating inference on the non-sampled portion of the population, on the other hand.  \citet{2009elliot} and \cite{2017elliot} assume that the reference sample size is sufficiently small that there is a negligible overlap in unit inclusions with the convenience sample.  This negligible overlap assumption is increasingly untenable under ever larger convenience samples. Later methods dispense with this assumption; in particular, \citet{savitskycombine2023} and \citet{2021valliant} make no assumption about the degree of overlaps in units to allow more robust inference.  Similarly, recent methods differ on how to estimate likelihoods specified for the population on realized (convenience and reference) samples.  \cite{2022Wu, 2021valliant} use a pseudo likelihood approach by approximating unknown population units with the weighted reference sample units.  The use of reference sample-weighted units may inflate estimation variance for small-sized reference samples. \citet{savitskycombine2023} directly specify a likelihood for the realized samples that avoids using reference sample weights.

To motivate the focus of our paper, we highlight a key covariate balance requirement of these methods to produce combined reference and convenience sample domain estimators with reduced bias (as compared to domain estimators obtained from solely using the reference sample).  

Quasi-randomization methods require availability of the covariates used to determine the sampling design (governing the reference sample) for convenience sample units. This requirement is generally readily satisfied for sampling designs parameterized by demographic variables; for example, in the case of surveys conducted from business establishments by the U.S. Bureau of Labor Statistics these covariates might include a discretized employment size class, industry classification and metropolitan statistical area designation.

\cite{2020Valliant} further notes that the target population units are assumed to have positive probabilities to be included into both samples conditional on the shared set of covariates among both reference and convenience samples. They refer to this condition of positive participation probabilties for all units in both samples as a requirement for ``common support ''. Satisfying common support requires that the support of covariate values expressed by units in the population is also expressed by units included in \emph{both} the reference and convenience samples. This paper addresses estimation bias that arises when common support is satisfied but where a subset of population units selected into the reference sample with relatively moderate-to-large inclusion probabilities may express vanishingly low convenience sample participation probabilities. Heuristically, there are often subsets of the population purposefully emphasized in the reference sample that are poorly represented in the convenience sample. 

Since a convenience sample derives from an opt-in or self-initiated participation process there will typically be some units in the realized convenience sample that are very different from those represented in the randomized reference sample.  To be precise, there may be some units in the convenience samples whose covariate values don't well overlap those for the reference sample.  \citet{GelmanHill:2007} discuss degrees of ``partial overlap'' in the space of covariate values that may occur between treatment and control sample arms in the causal inference experimental set-up and the increase in bias and variance in the resulting propensity scores. The low overlap of covariate values for those convenience units with the reference sample provides less information to estimate associated participation probabilities for them, which produces estimates with large errors.  Including these low overlap convenience units along with reference units to formulate a domain estimator would be expected to inflate bias and variance rather than reduce it.  The error inflating effect of these low overlap convenience units on the domain estimator would partially offset the variance reduction benefit of incorporating high overlap convenience units along with the reference units discussed in \citet{savitskycombine2023}.

This paper introduces an approach to identify and exclude a subset of convenience sample units whose covariate values poorly overlap the reference sample in order to further reduce the error in domain estimators that incorporate convenience units (and their estimated participation probabilities).  Our approach for excluding or thresholding units uses estimated reference and convenience sample inclusion and participation probabilities for the \emph{convenience} units as a uni-dimensional summary of the overlap of multivariate covariate values.   In the sequel we develop a set of alternative statistics used for thresholding where each statistic represents distinct functional combinations of the estimated reference and convenience sample inclusion and participation probabilities for the convenience units.  We note that \citet{savitskycombine2023} specify a Bayesian modeling approach that provides estimates both convenience \emph{and} reference sample participation and inclusion probabilities for the convenience units.  The most simple example of using these estimated probabilities to threshold units would be to exclude convenience units with low reference sample inclusion probabilities below some threshold quantile.  The logic for such a thresholding statistic is that convenience units with low values for estimated reference sample inclusion probabilities may be expected to express a low degree of overlap in covariate values with the reference sample.

We introduce a thresholding statistic for excluding convenience sample units that arises by minimizing of the variance of a domain mean estimator that is a function of the estimated reference and convenience sample inclusion and participation probabilities for the convenience sample units in Section~\ref{sec:optimal}. We begin by deriving the variance optimal thresholding statistic under the simpler set-up that composes the domain mean estimator using solely estimated convenience sample inclusion probabilities for convenience units (and excludes estimated reference sample inclusion probabilities for the convenience units).  We then derive our main result under a set-up that constructs a threshold statistic composed of both estimated reference and convenience sample marginal probabilities for the convenience units.  Section~\ref{sec:additional} introduces an additional thresholding statistic motivated by \citet{beresovsky2024review}.  We compare the reductions in bias and means squared error offered by the alternative thresholding statistics with a Monte Carlo simulation study in Section~\ref{sec:simulation} and conclude with a discussion in Section~\ref{sec:discussion}.

\section{Optimal Variance Thresholding}\label{sec:optimal}
\subsection{Thresholding based solely on convenience sample probabilities}
We begin this section using only convenience sample participation probabilities (obtained from co-modeling with the reference sample) for convenience units to construct our estimator to introduce our notation under a simpler thresholding construction.  This set-up contrasts with use of \emph{both} estimated convenience and reference participation and inclusion probabilities for the convenience units to compose our domain mean estimator.  We label the set-up that utilizes solely convenience sample participation probabilities (for convenience sample units) to define our thresholding statistic and set as ``one-arm''.  By contrast, our main result will use the more general set-up that defines the thresholding statistic from both estimated convenience and reference sample probabilities, which we label as ``two-arm". 

Our main result defines a set subset of $x \in \mathbb{X}$ where units in the convenience sample whose threshold statistic percentile (as a function of $x$) is less than a some small value ($\alpha$) will be excluded from the subset. Only convenience sample units that are members of the subset will be used to render our weighted domain mean estimator, $\hat{\mu}$.  

Let $\delta_{c}\in\{0,1\}$ index unit participation in the convenience sample where $\delta_{c} = 1$ denotes participation in the sample and $\delta_{c}  = 0$ denotes a non-participating unit from the population frame, $U$, where $\lvert U \rvert = N$.   Define marginal participation probability $\pi_{c}(x) = \Pr[\delta_{c} = 1\mid X=x]$ where $X \in \mathbb{X}$ is a random variable. This construction for $\pi_{c}(x)$ defines a marginal participation probability (rather than a propensity score).  We proceed to extend and adapt a result of \citet{11514} from the literature on  causal inference that defines a threshold statistic and acceptance set for units constructed from a subset of $x \in \mathbb{X}$ where the value of the threshold statistic is exceeded.  The acceptance set formed by excluding units whose value lies below some percentile of the threshold statistic constructed by \citet{11514} is guaranteed to produce a minimizing variance for the domain mean estimator after excluding those $x$ not in the acceptance set.  We repurpose and extend their result from treatment and control arms under their causal inference set-up to reference and convenience sampling arms under our survey sampling set-up.  We begin our extension of their result with a simpler result that defines an acceptance set and formulation for a thresholding statistic for units in a convenience sample that produces a minimum variance for the domain mean estimator constructed solely from convenience sample participation probabilities.

Our population quantity of inferential interest is $\mu = \mathbb{E}(Y)$ where $Y$ denotes a univariate response variable of interest.   Define our domain mean estimator as,
\begin{equation}
\hat{\mu} = \mu + \frac{1}{N}\mathop{\sum}_{i=1}^{N} \frac{z_{i}\delta_{i}}{\hat{\pi}_{c}(x_{i})},
\end{equation}
where we are assuming $N$ is known and $z = y - \mu$.   Treating $N$ as known may be relaxed, in practice. Let
\begin{align}
\phi(Y,\delta,X,\mu,e) &= \frac{z\delta}{\pi_{c}(X)}.\\
\hat{\mu} & = \mu + \frac{1}{N}\mathop{\sum}_{i=1}^{N}\phi(y_{i},\delta_{i},x_{i},\mu,e_{i})
\end{align}
Then $\phi(Y,\delta,X,\mu,e)$ has $0$ expectation and variance \citep[p. 1182]{db8ba95fe8f345adae86cf364ecfaa79}, 
\begin{align} 
    \mathbb{E}\left[\phi(Y,\delta,X,\mu,e)^{2}\right] &= \frac{1}{N}\mathbb{E}\left[\frac{\sigma_{1}^{2}(X)}{\pi_{c}(X)}\right] \label{eq:onearmvar},
\end{align}
where $\sigma_{1}^{2} = \mathbb{V}\left(Y\mid \delta = 1,X = x\right)$. The expectation on the LHS of Equation~\ref{eq:onearmvar}is taken with respect to the joint distribution for $X$ and the taking of a sample from the underlying frame on which $\mathbb{X}$ is defined.  The expectation on the RHS is taken with respect to the distribution for $X$.

Equation~\ref{eq:onearmvar} may be used in combination with Corollary 1 of \citet{11514} to produce the following result for the optimal threshold level, $\alpha$.  
\begin{thm}[One-arm extension of \citet{11514}]\label{thm:onearm}
Assume $\pi_{c}(x) > 0~\forall x  \in \mathcal{X}$
Then set $\mathbb{A} = \left\{x\in\mathbb{X}: \pi_{c}(x) > \alpha\right\}$ denotes the variance optimal subset of $\mathbb{X}$ after thresholding units where $\mathbb{A}$ is defined based on thresholding conditional inclusion probability, $\pi_{c}(X)$.   The minimum variance quantile $\alpha$ is constructed by,
\begin{equation}\label{eq:onearmrule}
\frac{1}{\alpha}  = 2\mathbb{E}\left[\frac{1}{\pi_{c}(X)} \bigg| \frac{1}{\pi_{c}(X)}<\frac{1}{\alpha}\right].
\end{equation}
For computation of $\alpha$ we approximate the expectation with sums over units $i \in S_{c}$, where $S_{c}$ denotes the observed convenience sample,
\begin{equation}\label{eq:computealpha}
    \frac{1}{\alpha} = 2\frac{\mathop{\sum}_{i\in S_{c}}\mathbf{1}(\hat{\pi}_{c}(x_{i})>\alpha)\frac{1}{\hat{\pi}_{c}(x_{i})}}{\mathop{\sum}_{i\in S_{c}}\mathbf{1}(\hat{\pi}_{c}(x_{i})>\alpha)}.
\end{equation}
\end{thm}
\begin{proof}
    Plugging in $\pi_{c}(X)$ for $e(X)$ into Theorem 1 of \citet{11514} and using the result of Equation~\ref{eq:onearmvar} for the case of where we utilize solely the convenience sample participation probabilities (for the convenience units) produces the result.
\end{proof}

\begin{remark}
The result of Theorem~\ref{thm:onearm} utilizes a one-arm set-up that composes the mean estimator from solely the convenience sample.  A companion, separate reference sample is required in order to estimate the convenience sample inclusion probabilities, $\hat{\pi}_{c}(x_{i}),~i \in (1,\ldots,N)$.  In the sequel, we will further extend Theorem~\ref{thm:onearm} by additionally estimating the reference sample inclusion probabilities for the same convenience units, $\hat{\pi}_{r}(x_{i}),~i \in (1,\ldots,N)$ also using the reference sample inclusion probabilities estimated on the convenience units.  See~\citet{savitskycombine2023} for more details on estimating $\left(\hat{\pi}_{c}(x_{i}),\pi_{r}(x_{i})\right)$ (where subscript ``$r$'' denotes reference sample) for convenience sample units.  They specify a model for the observed membership indicator in the pooled sample, $\mathbf{1}_{z_{i}}$, which is set to $1$ if unit $i$ is included in the convenience sample and $0$ if the unit belongs to the reference sample. Units in the convenience and reference samples are ``stacked'', which allows for a unit included in the convenience sample to also be included in the reference sample without the requirement to \emph{know} the identity of that unit.  They utilize a Bayesian hierarchical modeling approach that specifies a Bernoulli likelihood for indicator $\mathbf{1}_{z_{i}}$ for all units in the pooled sample.  A likelihood term is also included for $\pi_{r}(X_{i})$ only for units in the observed reference sample (where $\pi_{r}(X_{i})$ is known) to borrow further modeling strength.  This modeling set-up of \citet{savitskycombine2023} may also be performed in the frequentist paradigm. The main advantage of the Bayesian approach is that it treats values $\pi_{r}(X_{i})$ for the \emph{convenience} sample as \emph{unknown} and allows their estimation in the model.  By contrast, in the frequentist set-up (see \citet{beresovsky2024review}) $\pi_{r}(X_{i})$ are assumed known for \emph{all} convenience and reference sample units. 
\end{remark}

\begin{remark}
In this one-arm case where the domain estimator is constructed solely from the estimated convenience sample inclusion probabilities, the resulting thresholding is performed on the convenience sample inclusion probabilities, $\pi_{c}(x_{i}),~i\in S_{c} \subset U$ (where $S_{c}$ denotes units in frame $U$ that participate in the convenience sample), without accounting for the estimation quality of $\pi_{c}(X)$. So, this is a traditional regularization approach used to stabilize the variance of a survey domain estimator by excluding units with extreme weight values.  This approach trades some small increase in bias for a large decrease in variance.  
\end{remark}

\begin{remark}
We include an alternative, direct derivation for the result of Theorem~\ref{thm:onearm} in an Appendix~\ref{a:proof} assuming Equation~\ref{eq:onearmvar} is everywhere differentiable (on $x \in \mathbb{X}$).  We also include an illustration to show that the result of the Theorem does, indeed, produce a minimum variance estimator for $\hat{\mu}$.
\end{remark}

Equation~\ref{eq:onearmvar} can now be generalized in the manner of Section $3.1$ of \citet{11514} to develop an alternative to their Theorem 1 and Corollary 1 under a composite estimator that includes both reference and convenience sample inclusion and participation probabilities.

\subsection{Thresholding using both reference and convenience sample probabilities}\label{sec:twoarm}
Let $\delta_{c}$ and $\delta_{r}$ denote random inclusion indicators (governed by a survey design distribution) for convenience and reference samples, respectively, and let $\pi_{c}(x) = \Pr[\delta_c = 1 \mid X=x]$ and similarly for $\pi_{r}$.  Define our estimator as,
\begin{equation}\label{eq:compest}
\hat{\mu} = \mu + \frac{1}{N}\mathop{\sum}_{i=1}^{N} \frac{z_{i}\delta_{ci}}{\hat{\pi}_{c}(x_{i})} + \frac{z_{i}\delta_{ri}}{\pi_{r}(x_{i})},
\end{equation}

Although the above estimator is defined disjointly on the reference sample using $\pi_r(X)$ and the convenience sample using $\hat{\pi}_{c}(X)$, the resulting optimal variance thresholding rule of Equation~\ref{eq:threshrule} applies to \emph{only} units in the convenience sample.  So, as mentioned in Remark 4, below, we may use estimated $\hat{\pi}_{c}(x_{i})$ \emph{and} $\hat{\pi}_{r}(x_{i})$ for \emph{each} unit $i \in S_{c}$ to apply the thresholding rule of Equation~\ref{eq:threshrule}.  To demonstrate that this trick works, we may generate an estimator identical to Equation~\ref{eq:compest} that includes both convenience and reference sample probabilities defined \emph{solely} for convenience units.   Use $\{\pi_{c}(x_{i})\}_{i\in S_{c}}$ to generate a pseudo population of size $N$ (from units $i \in S_{c}$, allowing for replicates). Next take a random / probability sample from this pseudo population using $\{\pi_{r}(x_{i})\}$ of the same size as the reference sample.   Now form the same estimator as Equation~\ref{eq:compest}, but the universe of units is actually confined to $i \in S_{c}$.

Let
\begin{align}
\phi(Y,\delta_{c},\delta_{r},X,\mu,e_{c},e_{r}) &= \frac{z\delta_{c}}{\pi_{c}(X)} + \frac{z\delta_{r}}{\pi_{r}(X)}\\
\hat{\mu} & = \mu + \frac{1}{N}\mathop{\sum}_{i=1}^{N}\phi\left(y_{i},\delta_{ci},\delta_{ri},x_{i},\mu,\pi_{c}(x_{i}),\pi_{r}(x_{i})\right).
\end{align}

Then, from \citet{db8ba95fe8f345adae86cf364ecfaa79} the variance of our estimator is 
\begin{align}
    \mathbb{E}\left[\phi(Y,\delta,X,\mu,e)^{2}\right] &= \frac{1}{N}\mathbb{E}\left[\frac{\sigma_{c}^{2}(X)}{\pi_{c}(X)} + \frac{\sigma_{r}^{2}(X)}{\pi_{r}(X)}\right]\label{eq:twoarmvar} ,
\end{align}
where $\sigma_{c}^{2} = \mathbb{V}\left(Y\mid \delta_{c} = 1,X = x\right)$ and similarly for $\sigma_{r}^{2}$. The expectation on the LHS of Equation~\ref{eq:onearmvar} is taken with respect to the joint distribution for $X$ and the taking of a sample from the underlying frame on which $\mathbb{X}$ is defined.  The expectation on the RHS is taken with respect to the distribution for $X$.  We have used the assumption of independence between the sampling arms with respect to the design distribution.

We may now use Equation~\ref{eq:twoarmvar} to extend and generalize Corollary 1 of \citet{11514} in the case where $\sigma_{c}^{2} = \sigma_{r}^{2} = \sigma^{2}$.  
\begin{thm}[Two-arm extension of \citet{11514}]\label{thm:twoarm}\ \\
Assume $(\pi_{c}(x) > 0,\pi_{r}(x) > 0), ~\forall x\in \mathbb{X}$. \\Then $\mathbb{A} = \left\{x \in \mathbb{X}: \sqrt{\pi_{r}(X)\pi_{c}(X)/(\pi_{r}(X) + \pi_{c}(X))} > \alpha\right\}$ defines the optimal subset of $\mathbb{X}$ where threshold $\alpha$ is obtained as a solution to,
\begin{equation} \label{eq:threshrule}
    \frac{1}{\alpha^{2}} = 2\mathbb{E}\left[\frac{1}{\pi_{c}(X)} + \frac{1}{\pi_{r}(X)}\bigg|\frac{1}{\pi_{c}(X)} + \frac{1}{\pi_{r}(X)} \leq \frac{1}{\alpha^{2}} \right].
\end{equation}
\end{thm}
\begin{proof}
    Plugging in $\pi_{c}(x)$ for $e(X)$ and $\pi_{r}(X)$ for $1-e(X)$ into Theorem 1 of \citet{11514} and using the result of Equation~\ref{eq:twoarmvar} for the case of where we utilize both the convenience sample and reference sample participation and inclusion probabilities (for the convenience units) produces the result.
\end{proof}

\begin{remark}
Defining variance optimal subset, $\mathbb{A}$, by thresholding \newline $\sqrt{\pi_{r}(x_{i})\pi_{c}(x_{i})/(\pi_{r}(x_{i}) + \pi_{c}(x_{i}))} > \alpha$ is a harmonic mean that tends to exclude units $i$ where $\pi_{r}(x_{i})$ is a very different value from $\pi_{c}(x_{i})$.  We may even better understand the behavior of this thresholding statistic by noting the result from \citet{beresovsky2024review} that $\Pr[i \in S_{c},i\in S_{r}\mid i \in S] = \pi_{ri}\pi_{ci}/(\pi_{ri}+\pi_{ci})$, where $S = S_{c} \bigotimes S_{r}$ denotes the pooled convenience and reference sample.   This result reveals that convenience units with low probabilities of being in \emph{both} the convenience and reference samples tend to be excluded.  This thresholding behavior matches intuition because units with low probabilities to appear in both samples will tend to have low overlaps in their covariate supports.  We further note that our derivation of this variance minimizing threshold statistic was done without explicit reference to this joint probability, which makes the concordance of the two expressions (for the thresholding statistic, on the one hand, and the joint probability of inclusion in both samples, on the other hand) to be quite fortuitous.  We label this thresholding statistic as ``balanced" because it favors inclusion of records for estimating the domain mean that have relatively high probabilities of participating in \emph{both} samples.
\end{remark}

\begin{remark}
This thresholding method can be used in practice solely directed to units $i \in S_{c}$ because we have both estimated $\left(\hat{\pi}_{c}(x_{i}),\hat{\pi}_{r}(x_{i})\right)$ available.
\end{remark}

\begin{remark}
Theorem~\ref{thm:twoarm} assumes both $(\pi_{r}(x),\pi_{c}(x))$ are \emph{known} for the convenience units when, in fact, they are estimated.   We explore the sensitivity to the performance of the variance minimizing thresholding statistic (for the domain mean) of this theorem to estimation uncertainty for $(\hat{\pi}_{r}(x),\hat{\pi}_{c}(x))$ in the simulation study to follow.
\end{remark}

\subsection{Thresholding statistic motivated by variance structure of model score function}\label{sec:additional}
Our derivation of the thresholding statistic of Section~\ref{sec:twoarm} treats $\pi_{c}(\mathbf{x})$ as known.  By contrast, \citet{beresovsky2024review} suppose a generalized linear model, $\text{logit}(\pi_{ci}(\boldsymbol{\beta}))=\boldsymbol{\beta^T}\mathbf{x}_i$, with a linear form under a logit link function for logistic regression.  They derive the variance of the domain mean, $\hat{\mu}$, that includes an additive term for variance of the score function, $S(\boldsymbol{\beta})$, which has two parts:
\begin{align*}
\Var[S(\boldsymbol{\beta})] &= 
\Var[S_c(\boldsymbol{\beta})]+\Var[S_r(\boldsymbol{\beta})]=:\mathbf{A}+\mathbf{D}\\
\boldsymbol{D} &= 
\Var_{d}\left[ \sum\nolimits_{S_r}{ {\frac{ g_i}{ 1 + g_i }}  \left(1-{{\pi }_{ci}} \right) {{\mathbf{x}}_{i}}} \right],   
\end{align*}
where $g_{i} = \pi_{c}(\mathbf{x}_{i})/\pi_{r}(\mathbf{x}_{i})$ and $\Var_{d}$ denote the design variance.  
Motivated by the dependence of $\boldsymbol{D}$ on $g_{i}$, we propose to use this statistic as another thresholding option.

 We propose the following acceptance set that uses $g$:
\begin{equation*}
  \mathbb{A} = \left\{x \in \mathbb{X}: \pi_{r}(x)/\pi_{c}(x) > \alpha\right\}.
\end{equation*}
\begin{remark}
    The use of $\pi_{r}(x)/\pi_{c}(x)$ as a thresholding statistic may be intuitively motivated by noting that it will tend to threshold or exclude units $i \in S_{c}$ where $\pi_{r}(\mathbf{x}_{i})$ is relatively small for each unit and $\pi_{c}(\mathbf{x}_{i})$ is relatively large, which may occur if the value for $\mathbf{x}_{i}$ for some $i \in S_{c}$ is not well covered by or represented in the reference sample, $S_{r}$.
\end{remark}

\section{Simulation study}\label{sec:simulation}

\subsection{Simulation design}
We conduct a Monte Carlo simulation study that generates a finite population on each iteration to include covariates $\mathbf{x}$ that govern both the convenience and reference sample designs.  The sample designs are size-based as a linear function of $\mathbf{x}$ where we vary the coefficients of the linear function to draw two categories of reference and convenience samples: 1. Where the covariate spaces of resulting reference and convenience samples express a \emph{high} degree of overlap; 2. Where the two samples express a \emph{low} degree of overlap.  We also generate a response variable of interest, $y$, for the finite population.  A domain mean, $\mu$, is constructed for the population and \emph{estimated} by a combined weighted estimator over the reference and convenience samples.  Finally, we compare the $3$ thresholding methods we developed in Section~\ref{sec:optimal} in terms of their bias, error and coverage performances.  We expect that conducting thresholding of sampled convenience units using one or more of our thresholding statistics will reduce estimation error.

We utilize the simulation data generation process of \citet{savitskycombine2023}. We briefly summarize the procedure and refer the reader for a more detailed exposition.  We generate $M = 30$ distinct populations, each of size $N = 4000$.  Design covariates, $X$, of dimension $K = 5$ are generated (all binary, with one continuous).   Outcome variable, $y_{i}$, is generated as $\log(y_i) \sim \mathcal{N}(\mathbf{x}_i \beta, 2)$ for $i = 1,\ldots, N$.

A randomized reference sample of size $n_{r} = 400$ is taken from the finite population under a proportion-to-size (PPS) design with size variable, $s_{r_i}  = \log(\exp(\mathbf{x}_i \times \beta) + 1)$.

For the convenience sample, we set $n_{c} \approx 800$, which is a relatively larger sampling fraction that we choose to explore the full range of $\pi_c \in [0,1]$ that we would expect to see for business establishment data in the U.S. Bureau of Labor Statistics.  We use a size-based Poisson sample with $\pi_{c_i} = \mbox{logit}^{-1}(\mathbf{x}_i \times \beta_{c} + \mbox{offset})$.  We control `high' and `low' overlap by varying $\beta_{c}$ compared to the reference sample.

Figure~\ref{fig:overlap} presents a violin (rotated and reflected density) plot for the percentage overlap of \emph{units} in both the convenience and reference samples over the Monte Carlo iterations.  The left-hand plot represents the high overlap samples and the right-hand plot represents the low-overlap samples.  We see that the number of shared units in both samples is notably higher for the high overlap samples than for the low overlap samples.  We expect fewer units to be thresholded for a high overlap sample since their covariate supports express relatively more overlap suct that units in the convenience sample are more similar to those in the reference sample.   Since our modeling obtains information about the population from the reference sample (and reference sample inclusion probabilities) we are able to better estimate participation probabilities for convenience units that are similar in covariate values to the reference units.

\begin{figure}
\centering
\includegraphics[width = 0.95\textwidth]{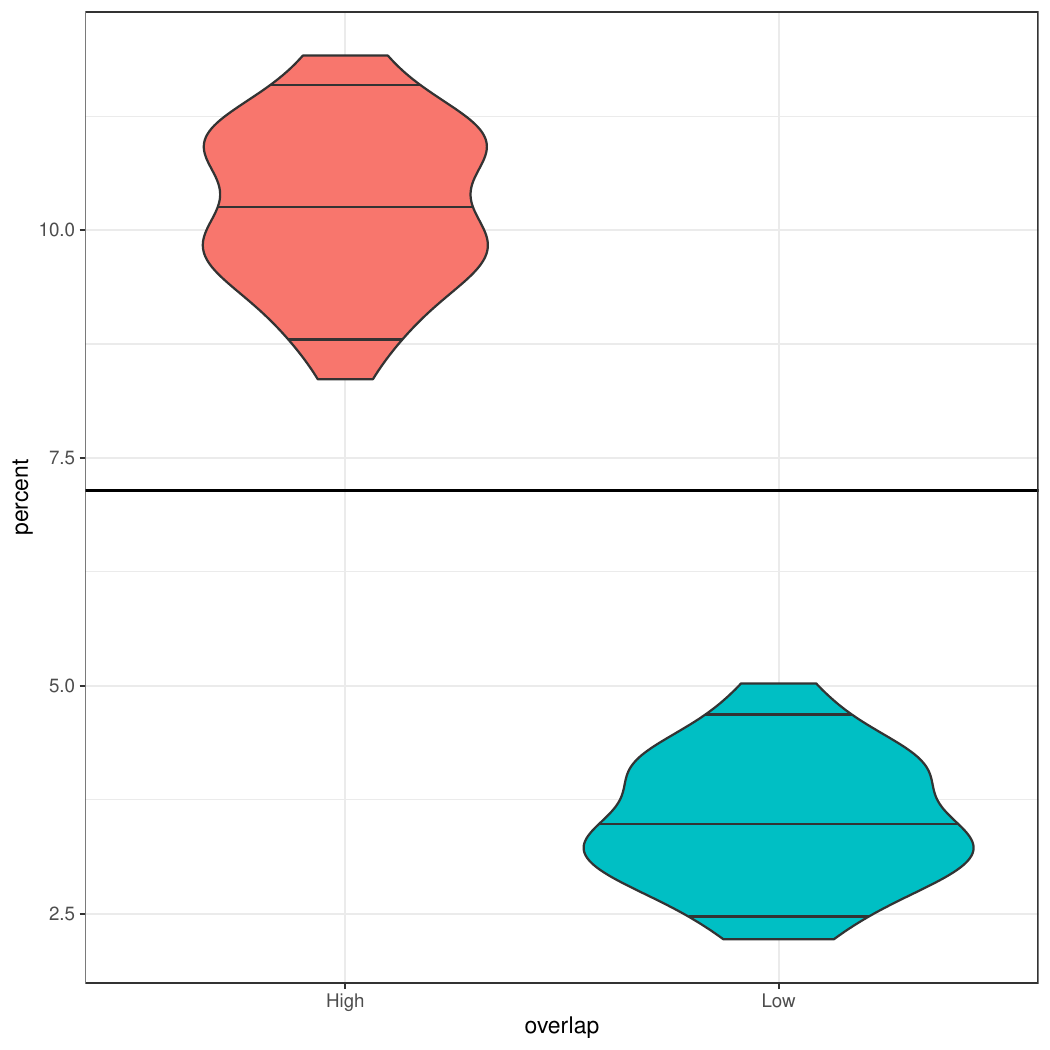}
\caption{Distribution over $M = 30$ Monte Carlo iterations of the percentage of units overlapping between realized reference and convenience samples (taken on each Monte Carlo iteration).
    }
\label{fig:overlap}
\end{figure}
\afterpage{\FloatBarrier}

\subsection{Thresholding of convenience units}

In this paper we employ the Bayesian model formulation of \citet{savitskycombine2023} that estimates both $(\pi_{r}(\mathbf{x}_{i}),\pi_{c}(\mathbf{x}_{i})),~i\in S_{c}$.  In the sequel we use $(\pi_{ri} = \pi(\mathbf{x}_{i}), \pi_{ci} = \pi_{c}(\mathbf{x}_{i}))$ for ease-of-reading and to emphasize the dependence on $i \in S_{c}$.

Within each Monte Carlo iteration, $m \in 1,\ldots,M$, we conduct thresholding of convenience units and computation of the domain mean for each posterior/MCMC sample in the following procedure:
\begin{enumerate}
    \item For \emph{each} posterior/MCMC draw $s \in 1,\ldots S$, compute the thresholding statistic (e.g., balanced thresholding statistic) for each unit $i \in S_{c}$ as a function of $(\hat{\pi}_{rsi},\hat{\pi}_{csi})$.  Denote the focus thresholding statistic as, $T(\hat{\pi}_{rsi},\hat{\pi}_{csi})$, that allows us to provide a general exposition of how we conduct thresholding of convenience units; for example, we may set $T(\hat{\pi}_{rsi},\hat{\pi}_{csi}) = \sqrt{\hat{\pi}_{rsi}\hat{\pi}_{csi}/(\hat{\pi}_{rsi} + \hat{\pi}_{csi})}$ for $i \in S_{c}$.
    \item For MCMC iteration $s$: evaluate the distribution of the thresholding statistic $T(\hat{\pi}_{rsi},\hat{\pi}_{csi})$ over the convenience units, $i \in S_{c}$, and compute threshold quantile, $\alpha_{s}$ associated with target percentile, $\gamma$, below which convenience units are excluded / thresholded.
    \item \emph{Retain/accept} those \emph{convenience} units where $\mathbb{A}_{s} = \{i\in S_{c}: T(\hat{\pi}_{rsi},\hat{\pi}_{csi}) > \alpha_{s}  \}$.
    \item Use the retained units in draw $s$ to construct the domain mean, $\mu_{s} = (\sum_{i\in \mathbb{A}_{s}}y_{i}/\hat{\pi}_{csi} + \sum_{i\in S_{r}}y_{i}/\hat{\pi}_{rsi})/ (\sum_{i\in \mathbb{A}_{s}}1/\hat{\pi}_{csi} + \sum_{i\in S_{r}}1/\hat{\pi}_{rsi})$.  
    \item  One now has the induced posterior distribution over the $S$ MCMC samples for $\mu$ from which one may estimate the mean (e.g., $\mu = 1/S\sum_{s=1}^{S}\mu_{s}$).
\end{enumerate}
\begin{remark}
    The above procedure is a form of ``soft'' thresholding because a unit $i \in S_{c}$ may be excluded on posterior sampling draw $s$ in forming domain mean estimator $\mu_{s}$, but then \emph{included} in posterior draw $s^{'}$ to construct $\mu_{s^{'}}$.  So each $\mu_{s}$ may be constructed from a differing set of convenience units. This occurs because $(\hat{\pi}_{rsi},\hat{\pi}_{csi})$ are parameters estimated from our model, so the distribution of $T(\hat{\pi}_{rsi},\hat{\pi}_{csi})$ over convenience units $i \in S_{c}$ will vary from over the posterior draws, $s \in 1,\ldots,S$.  
    
    We formulate a variation to this procedure that produces a ``hard'' threshold to compare the performance to our main soft thresholding procedure.  For the hard thresholding alternative, we construct the acceptance sets $\mathbb{A}_{s},~s = 1,\ldots,S$ as described in the first $3$ steps of the above procedure.  We then count the percentage of over the $S$ posterior draws that unit $i$ is in each acceptance set $\mathbb{A}_{s}$.  If the percentage is less than $50\%$ we \emph{exclude} or threshold unit $i$.  In other words, we form a single acceptance set over the $S$ MCMC draws with $\mathbb{A} = \{i\in S_{c}:i\in \mathbb{A}_{s} \text{ for a total of } S^{i} = \sum_{s}^{S}(1:i\in \mathbb{A}_{s}) > 0.5S\}$.   So, our first addtional steps formulates $\mathbb{A}$, the set of non-thresholded convenience units.  We then use this \emph{same} set of units to compute $\mu_{s}$ for each MCMC draw.  So, either unit $i$ is included to construct all the $\mu_{s}$ or it is excluded.  We use the label ``two-step" for this hard thresholding alternative since we first threshold the units over all MCMC draws and then compute the domain mean estimator.
\end{remark}

\begin{remark}
    Although our thresholding procedure is constructed under the Bayesian model formulation of \citet{savitskycombine2023} for developing a thresholded posterior distribution for domain mean, $\mu$, steps $1-4$ of our thresholding procedure may be applied under the frequentist generalized linear formulation of \citet{beresovsky2024review} to obtain a thresholded estimator of $\mu$ with no loss of generality or applicability.  Instead of thresholding each MCMC draw, $s$, one would threshold the statistic formed from the maximum likelihood estimators of the convenience sample participation probabilities under frequentist model estimation.
\end{remark}

\subsection{Results}

Figure~\ref{fig:two_arm_results} presents plot panels for bias, root mean squared error (RMSE), median absolute deviation (MAD) and coverage results over the $M$ Monte Carlo iterations.  The left side of each horizontal bar in the plot panels represents a result for ``L'' or the low overlap sample, while the right side of each horizontal bar represents a result for ``H'' or the high overlap sample. 
The top most row of bars in blue presents results using the unknown \emph{true} values for both the reference sample inclusion probabilities for the reference sample units and the convenience sample inclusion probabilities the convenience units as if they were known.
The next row of bars down from the top in red presents the result from the model of \citet{savitskycombine2023} that smooths or co-models the inclusion probabilities for the reference sample units.   No thresholding is conducted for the results in these first two rows. 
The next two rows of bars present results for our variance optimal balanced threshold statistic:  the orange bar uses our main soft thresholding procedure, while the yellow bar uses the alternative hard thresholding procedure that we label as ``two-step".
The next row of light green bars presents results for thresholding $\pi_{ri}$ while the last row of green bars presents results for the ratio ($\pi_{ri}/\pi_{ci}$) thresholding statistic.  We remind the reader that the statistics and thresholding are performed over $i \in S_{c}$ (the convenience sample) and that our Bayesian model estimates both ($\pi_{ri},\pi_{ci}$) for each unit in the convenience sample.   
The vertical black dashed line in each plot panel represents the result using \emph{only} the reference sample (and excluding the convenience sample).  We use the $\gamma = 5\%$ of the distribution over the convenience for each threshold statistic to compute the thresholding quantile, $\alpha$.

One notes that the estimation errors (RMSE, MAD) are little different both with and without thresholding and among the thresholding statistics for the high (H) overlap samples, which is expected because there is less need for thresholding due to the high degree of overlap in covariate spaces between the reference and convenience samples such that most convenience sample inclusion probabilities are well-estimated.   By contrast, we observe that the estimation errors for the balanced statistic perform best among the different thresholding statistics and even better than the case where use the true convenience sample participation probabilities (blue bars) as is they were known.  The slight increase in bias relative to the blue bar is more than offset by a decrease in variance, producing lower estimation error.  There is little difference between the soft and hard thresholding alternatives under the balanced statistic, though the soft thresholding produces a slightly higher amount of bias but also a slightly lower amount of estimation error as compared to hard thresholding.  Perhaps we are not surprised that the balanced threshold statistic performed best because it was derived as a minimum variance estimator for the domain mean, though it is surprising that this thresholding option performed better for low overlap (L) samples than did the domain mean estimator constructed from the true (rather than estimated) convenience sample inclusion probabilities (as if they were known).  

Lastly, while the balanced threshold statistic produces only a slight improvement in error for high overlap (H) samples, the notion of whether a convenience sample is high or low overlap is relative such that the practitioner may not know whether their realized reference and convenience samples represent H or L.  Nevertheless, since thresholding with the balanced statistic never produces worse errors than not thresholding and sometimes much better there is little risk to use thresholding.

\begin{figure}
\centering
\includegraphics[width = 1.0\textwidth]{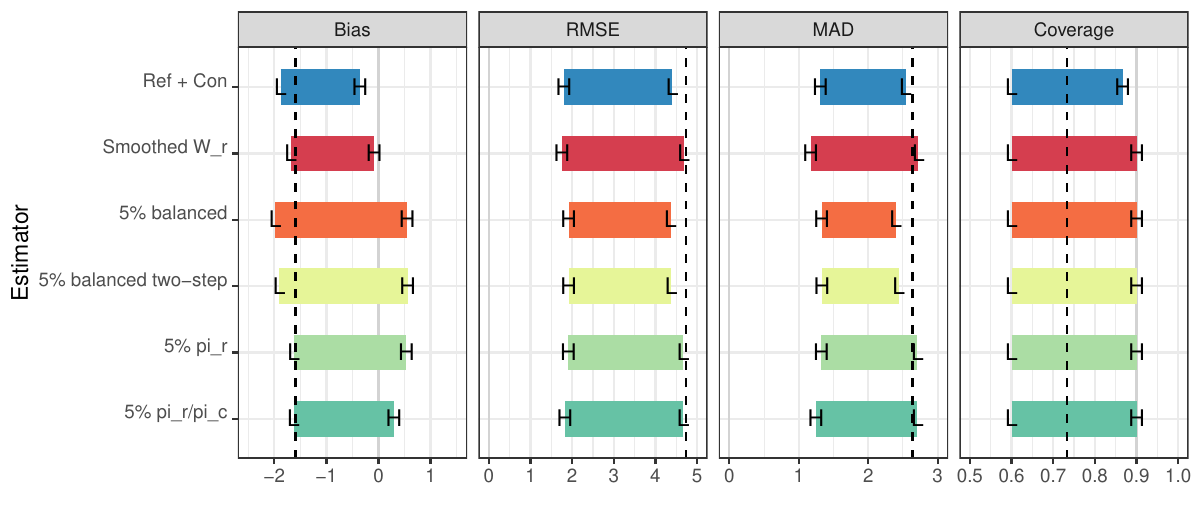}
\caption{Performance of the weighted mean estimator between high (H) and low (L) overlapping samples using variations of the two-arm method across Monte Carlo Simulations for (top to bottom): True weights for both samples (Blue), Smoothed weights for reference sample (Red), minimum variance or balanced $\sqrt{\pi_{r}(x)\pi_{c}(x)/(\pi_{r}(x) + \pi_{c}(x))}$ (Orange), balanced based on posterior mean (Yellow), $\pi_r$ only (Light Green), $\pi_r/\pi_c$ (Dark Green). Left to right: Bias, root mean square error, mean absolute deviation, coverage of 90\% intervals. Vertical reference line corresponds to using the reference sample only.
    }
\label{fig:two_arm_results}
\end{figure}
\afterpage{\FloatBarrier}

We chose a reasonably small ($5\%$) percentile for thresholding, so we next experiment with $10\%$ and $1\%$ under our best performing balanced thresholding statistic (under soft thresholding).  Figure~\ref{fig:two_arm_variations} presents the results.  While the estimation errors are similar for the $3$ different percentiles for low overlap (L) samples, we nevertheless note that the error performance is notably better for $1\%$ balanced thresholding under high overlap (H) samples than the other two higher thresholding percentiles, and even performs slightly better than the blue bar that uses true convenience sample participation probabilities.  Thresholding fewer units for high overlap samples intuitively makes sense since convenience units are relatively more similar to reference units.  The low overlap sample MAD is, however, worst for the $1\%$ threshold and best for the $10\%$ threshold, which also accords with intuition since the convenience units in low overlap samples are less similar (in their covariate values) to reference sample units.   Yet, the worsening of estimation error in the low overlap is a much smaller magnitude than the improvement in estimation error for high overlap.   Our results suggest that the practitioner may generally favor a relatively lower value for the thresholding percentile.

While thresholding does notably reduce estimation errors (RMSE/MAD) on low overlap samples, as expected, uncertainty quantification is little improved (and continues to express undercoverage) even after thresholding due to the limited estimation improvement offered for a low overlap convenience sample.  The fidelity of uncertainty quantification is driven by the underlying degree of overlap in the covariate supports of the reference and convenience sampling arms and is not much affected by thresholding relatively few convenience units. As a result of the low quality of uncertainty quantification under the low overlap samples, the coverage performances for all methods express little differentiation.  By contrast, for high overlap the coverage results are more robust and nominal coverage is achieved when thresholding relatively fewer units, as expected.  Thresholding is most important for low overlap samples to prevent non-representative outliers from inducing large errors (due to biased estimation of their convenience sample inclusion probabilities).  Our results show that thresholding for low overlap samples provides a notable improvement in error control over repeated sampling.

\begin{figure}
\centering
\includegraphics[width = 1.0\textwidth, page = 2]{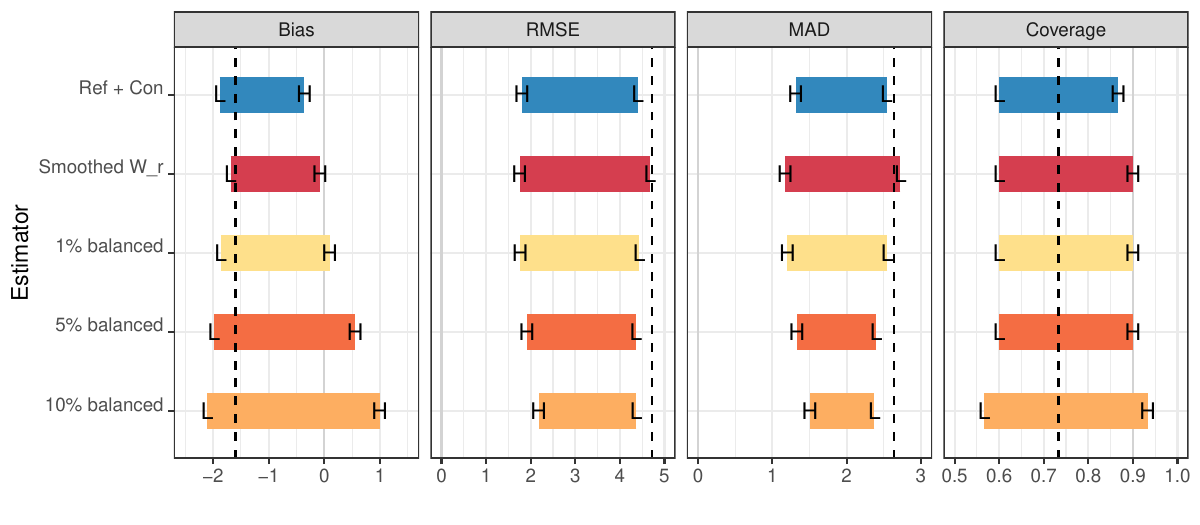}
\caption{Comparison of the variance for the balanced threshold, $\sqrt{\pi_{r}(x)\pi_{c}(x)/(\pi_{r}(x) + \pi_{c}(x))}$ between high (H) and low (L) overlapping samples for (top to bottom): True weights for both samples (Blue), Smoothed weights for reference sample (Red), 1\% (Yellow) vs. 5\% (Orange) and 10\% (Light Orange).  Left to right: Bias, root mean square error, mean absolute deviation, coverage of 90\% intervals. Vertical reference line corresponds to using the reference sample only.
    }
\label{fig:two_arm_variations}
\end{figure}
\afterpage{\FloatBarrier}

We recall that the balanced threshold statistic was derived to produce a minimum variance domain mean estimator.  Yet, the result in Section~\ref{sec:optimal} assumes that the reference sample inclusion and convenience sample participation probabilities for convenience sample units, $(\pi_{ri},\pi_{ci}),~i\in S_{c}$, are \emph{known} when, in fact, they are estimated.   We seek to assess the sensitivity of the thresholding statistics to uncertainty in estimation of these inclusion and participation probabilities for convenience units.  

Each curve in a each plot panel of Figure~\ref{fig:uncertainty} presents a sequence of $90\%$ credibility intervals of percentiles for the fit statistic estimated on each MCMC iteration.  More specifically, if we fix an MCMC iteration, we next compute the estimated thresholding statistic from the probabilities for each unit and compute its percentile of the distribution of the statistic over the convenience sample units.   We repeat this process for each MCMC draw, which gives us a range of percentiles of the thresholding statistic for each convenience sample unit.  Each horizontal line in the curve represents the $90\%$ credibility interval of the percentiles for a convenience sample unit.  These lines are ordered along the horizontal axis by the posterior mean of estimated thresholding statistic for each unit.  The longer the horizontal lines, the greater the estimation uncertainty for the thresholding statistic.   The blue-colored horizontal lines represent those units who have switched from being on one side of threshold to the other (meaning, they were sometimes included and sometimes excluded) more than $10\%$ of the MCMC samples.  The horizontal dashed lines in each panel represent $1\%, 5\%, 10\%$ thresholds (from bottom-to-top).

The left-hand curve in each plot panel represents estimations under low overlap samples and the right-hand represents high overlap samples.  The left plot panel represents the the balanced thresholding statistic, while the right panel represents the ratio thresholding statistic. 

Focusing on the left-hand panel for the balanced thresholding statistic, we see that the relatively wider horizontal lines for the low overlap sample express more estimation uncertainty than do those for the high overlap sample.  That accords with our expectation because the reference sample provides less information about convenience units whose covariate values are different from those of the reference sample.  Yet, we see relatively few units (colored in blue) that switch between being excluded/thresholded and included for estimating the domain mean.   So, the uncertainty does not impact the thresholding set and that explains why the balanced thresholding statistic turned out to be variance optimal as compared to the other thresholding statistics despite the uncertainties in estimating inclusion and participation probabilities.  By contrast, we observe a relatively higher number of units that switch between inclusion and exclusion under the ratio thresholding statistic in the right-hand plot panel.  So, the performance of this thresholding statistic is less robust under uncertainty about the probabilties than is the balanced thresholding statistic.

\begin{figure}%
    \centering
    \subfloat[\centering Balanced Thresholding]{{\includegraphics[width=5cm]{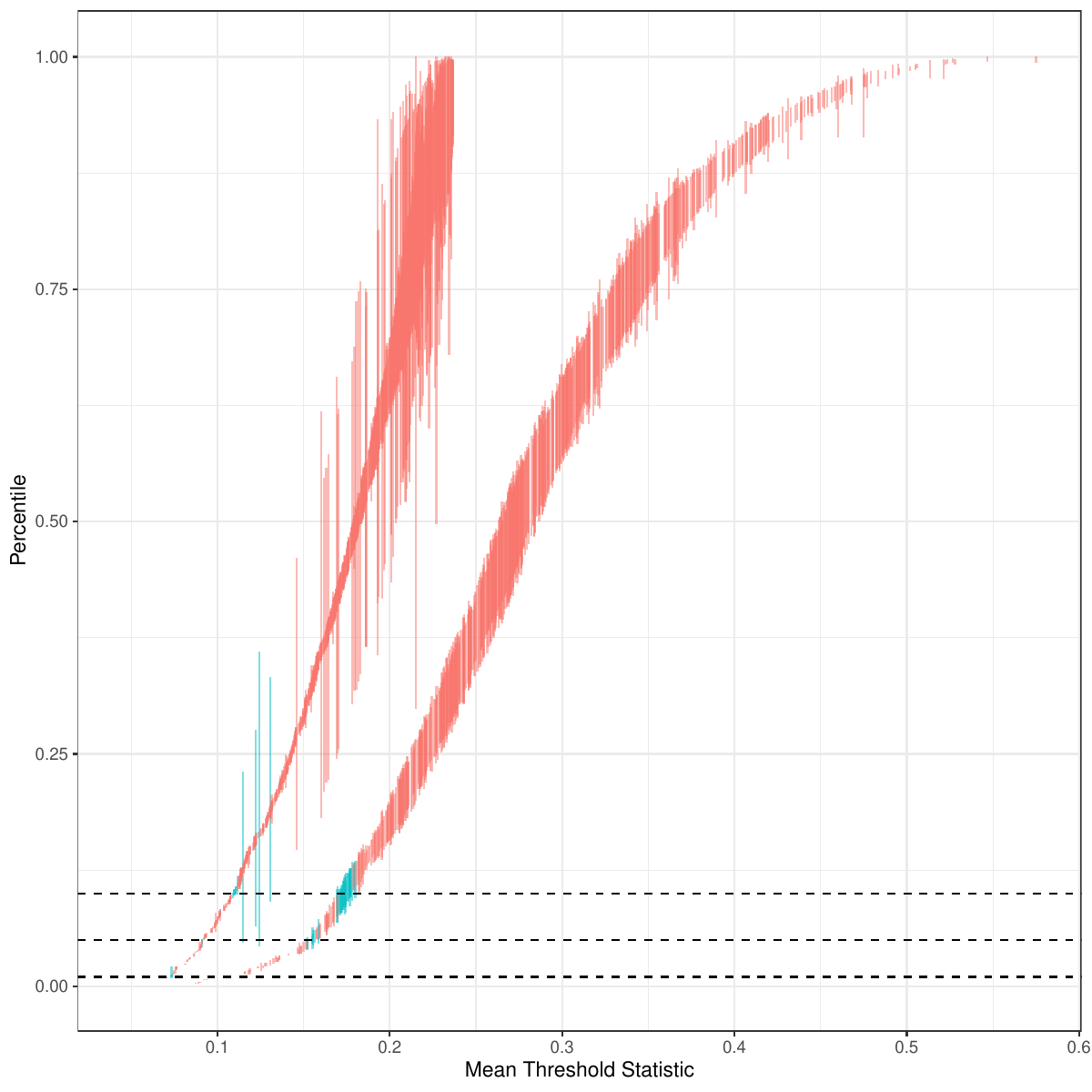} }}%
    \qquad
    \subfloat[\centering Ratio Thresholding]{{\includegraphics[width=5cm]{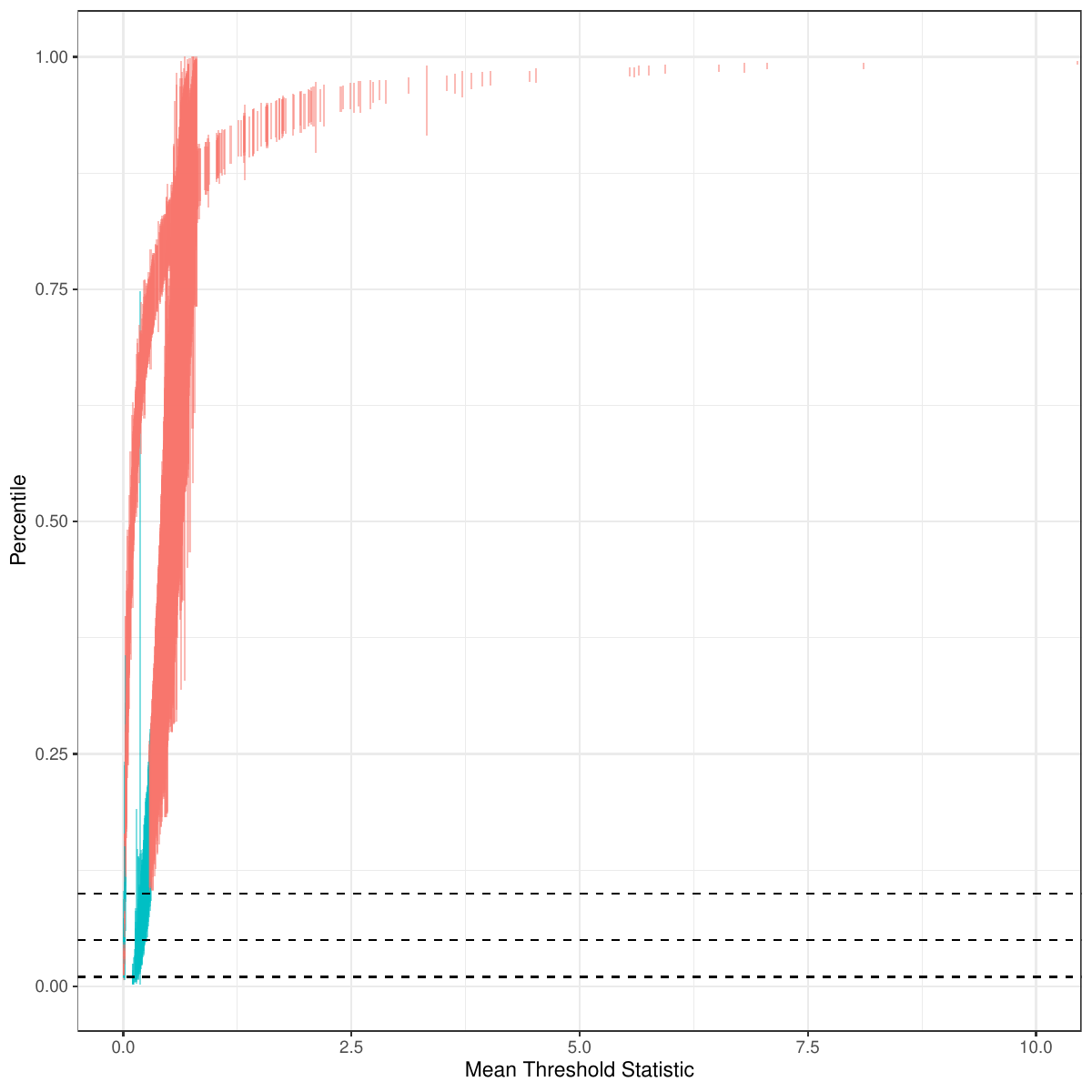} }}%
    \caption{Vertical lines are percentiles of Threshold statistic distribution over 700 MCMC draws of $\pi_{ci},\pi_{ri}$ for convenience units. Left is L and Right is H.  Blue denotes unit that jumped threshold $> 10\%$ of draws}%
    \label{fig:uncertainty}%
\end{figure}
\afterpage{\FloatBarrier}

\section{Discussion}\label{sec:discussion}
The quasi-randomization method of \citet{savitskycombine2023} that treats the non-randomized convenience sample as if it arose from a latent survey design process with an unknown sampling distribution provides a start-of-art method for producing survey-weighted domain estimates.   Yet, the estimation quality of inclusion and participation probabilities for convenience units depends on the degree of overlap in the design covariate spaces between the randomized reference and convenience samples.   It is typically the case that the estimated convenience sample participation probabilities for some convenience units whose design covariate values are very different from the reference sample are not well-estimated.  Incorporating these units can partially defeat the purpose of leveraging the convenience sample by actually increasing bias and variance as compared to excluding them.

We devised a soft thresholding procedure for excluding convenience sample units that are very different from reference sample units and achieved a notable reduction in estimation error for low overlap (in their design covariate spaces) samples.  We began by developing a new formulation for a balanced threshold statistic that minimized the resulting variance of the domain estimator.  Our balanced thresholding statistic proposes to exclude some convenience sample units and is constructed from inclusion and participation probabilities for convenience units that effectively serve as one-dimensional summaries of the design covariates.  It was particularly interesting to discover that the balanced threshold statistic derived from a theoretical exposition turns out to be a function of the joint probability that a unit is in \emph{both} the reference and convenience samples.  This formulation makes intuitive sense because our procedure proposes to exclude those convenience units that express low probabilities of being in both samples.

We motivated an additional thresholding statistic that we labeled ``ratio" as the ratio of reference and convenience sample inclusion probabilities based on the variance formulation of the domain mean estimator derived in \citet{beresovsky2024review}.

We designed a soft thresholding procedure that constructed an acceptance set for convenience units to be included in domain mean estimator on each MCMC iteration such that a unit might be included in some iterations but not others.   

Our result revealed that the balanced threshold statistic produced the greatest reduction in the variance of the domain estimator, particularly for relatively lower overlap samples.  We also showed that this reduction is relatively insensitive to the percentile cutoff for the estimated distribution of the balanced threshold statistic over the convenience sample units.   Finally, we showed that this variance reduction result is robust against estimation uncertainty because the units that are thresholded are minimally impacted under our soft thresholding procedure.

\appendix

\section{Direct derivation of variance minimizing threshold for one-arm sample }\label{a:proof}

The Hajek mean estimator from the convenience sample $S_c$ is:
\begin{align*}
\hat{\bar{y}}=\frac{\sum\nolimits_{{{S}_{c}}}{\frac{y\left( x \right)}{\hat{e}\left( x \right)}}}{\sum\nolimits_{{{S}_{c}}}{\frac{1}{\hat{e}\left( x \right)}}}
\end{align*}
where $\hat{e}\left( x \right)$ is estimated propensity score. \\
The associated model-based variance of this estimator is:
\begin{align*}
\operatorname{var}\left( {\hat{\bar{y}}} \right)=\frac{\sum\nolimits_{{{S}_{c}}}{\frac{\sigma _{y}^{2}\left( x \right)}{{{{\hat{e}}}^{2}}\left( x \right)}}}{{{\left[ \sum\nolimits_{{{S}_{c}}}{\frac{1}{\hat{e}\left( x \right)}} \right]}^{2}}}
\end{align*}

Assume that all variance  $\sigma _{y}^{2}\left( x \right)=\sigma _{y}^{2}$ are equal. Order convenience sample units by response propensity $\widehat{e}\left( x \right)$. Units can be listed by $\widehat{e}\left( x \right)$ with density $w(\widehat{e}\left( x \right))=  \widehat{e}\left( x \right) $. Variance estimated from full convenience sample $S_c$ without cut-off may be expressed as integral over the distribution of response propensity $\widehat{e}\left( x \right)$

\begin{align*}
\operatorname{var}\left( {\hat{\bar{y}}} \right)=\frac{\int_{0}^{1}{\frac{\sigma _{y}^{2}\left( x \right)}{{{{\hat{e}}}^{2}}\left( x \right)}w\left( \hat{e}\left( x \right) \right)d\left( \hat{e}\left( x \right) \right)}}{{{\left[ \int_{0}^{1}{\frac{1}{\hat{e}\left( x \right)}w\left( \hat{e}\left( x \right) \right)d\left( \hat{e}\left( x \right) \right)} \right]}^{2}}}=
\frac{\sigma _{y}^{2}\int_{0}^{1}{\frac{1}{{{{\hat{e}}}}\left( x \right)}d\left( \hat{e}\left( x \right) \right)}}{{{\left[ \int_{0}^{1}{d\left( \hat{e}\left( x \right) \right)} \right]}^{2}}}
\end{align*}

If sample units are trimmed by response propensity at level $\varepsilon$, then variance depending on $\varepsilon$ is
\begin{align*}
\operatorname{var}\left( \hat{\bar{y}},\varepsilon  \right)=
\frac{\sigma _{y}^{2} \int_{\varepsilon }^{1}{\frac{1}{{{{\hat{e}}}}\left( x \right)}d\left( \hat{e}\left( x \right) \right)}}{{{\left[ \int_{\varepsilon }^{1}{d\left( \hat{e}\left( x \right) \right)} \right]}^{2}}}=
\frac{\sigma _{y}^{2} F\left( \varepsilon  \right)}{{{G}^{2}}\left( \varepsilon  \right)},
\end{align*}
where $F(\hat{e}(x))$ is a primitive of $f(\hat{e}(x)) = 1/\hat{e}(x)$ and $G(\hat{e}(x))$ is a primitive of $1$.

Minimize the trimmed variance by $\varepsilon$
\begin{align*}
\frac{d\operatorname{var}\left( \hat{\bar{y}},\varepsilon  \right)}{d\varepsilon }=
\frac{\sigma _{y}^{2} F'\left( \varepsilon  \right){{G}^{2}}\left( \varepsilon  \right)-2G'\left( \varepsilon  \right)G\left( \varepsilon  \right)\sigma _{y}^{2} F\left( \varepsilon  \right)}{{{G}^{4}}\left( \varepsilon  \right)}=0
\end{align*}

Here we have:
\begin{align*}
 F'\left( \varepsilon  \right) &= \frac{d}{d\varepsilon}\left(F(1) - F(\varepsilon)\right) = 0 - \frac{1}{\varepsilon}\times 1\\
 G'\left(\varepsilon\right) &= \frac{d}{d\varepsilon}\left(G(1) - G(\varepsilon)\right) = G'(1) - G'(\varepsilon) = -1.
\end{align*}

Optimal propensity cut-off point $\varepsilon$ can be estimated from the numerator null condition
\begin{align*}
\frac{1}{\varepsilon_c} G \left( \varepsilon_c \right)-2F\left( {{\varepsilon }_{c}} \right)=0 \\
\frac{1}{\varepsilon_c} = \frac{2 F\left( {{\varepsilon }_{c}} \right)}{G \left( \varepsilon_c \right)} =
\frac{2\sum\nolimits_{{{S}_{c}}}{\frac{1}{{{{\widehat{e}}}}\left( x \right)}}\left| \widehat{e}\left( x \right)>{{\varepsilon }_{c}} \right.} {\sum\nolimits_{{{S}_{c}}}{1\left| \hat{e}\left( x \right)>{{\varepsilon }_{c}} \right.}}
\end{align*}

Results of simulations:
\begin{itemize}
    \item Sample size $n=1,400$
    \item Propensity score $\widehat{e} \sim Beta(1,2)$
\end{itemize}

\begin{figure}
\centering
  \includegraphics[width=1.0\linewidth]{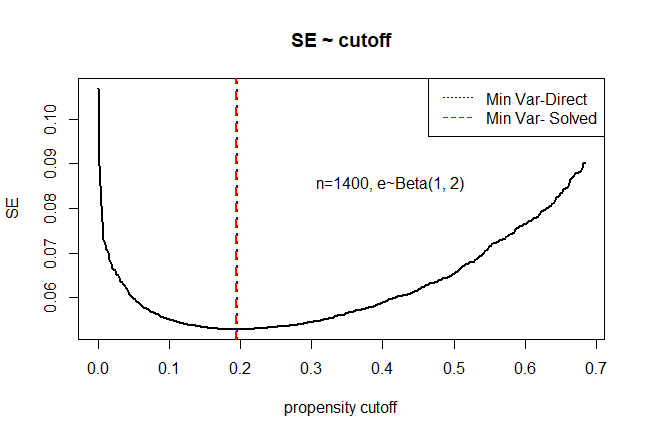}
\end{figure}


\bibliography{ref}

\end{document}